\documentclass[12pt]{amsart}
\usepackage{graphicx}
\newtheorem{theorem}{Theorem}[section]
\newtheorem{corollary}[theorem]{Corollary}

\newtheorem{remark}[theorem]{Remark}
\newtheorem{proposition}[theorem]{Proposition}

\newtheorem{example}[theorem]{Example}
\title[Optimal execution with geometric price process]{Optimal execution of equity with geometric price process}
\author[G. Hernandez-del-Valle]{Gerardo Hernandez-del-Valle}
\address{Statistics Department, Columbia University\\1255 Amsterdam Ave. Room 1005, New York, N.Y., 10027.}
\email{gerardo@stat.columbia.edu}
\thanks{The research of this author was partially supported by Algorithmic Trading Management LLC}
\author[C. G. Pacheco-Gonzalez]{Carlos G. Pacheco-Gonzalez}
\address{Mathematics Department, CINVESTAV-IPN\\Av. Instituto PolitŽcnico Nacional \#2508, Col. San Pedro Zacatenco, M\'exico, D.F., C.P. 07360}
\email{cpacheco@math.cinvestav.mx}
\date{October 14, 2009}
\begin{document}
\maketitle
\begin{abstract}
In this paper we derive the  Markowitz-optimal, deter\-ministic-execution trajectory  for a trader who wishes to buy or sell a large position  of a share  which evolves as a geometric Brownian motion in contrast to the arithmetic model which prevails in the existing literature. Our calculations include a general temporary impact, rather than a specific function. Additionally,  we point out---under our setting---what are the necessary ingredients to tackle the problem with adaptive execution trajectories.   We provide a couple  of examples which illustrate the results. We would like to stress the fact that in this paper we use understandable user-friendly techniques.
\end{abstract}
\section{Introduction}
The problem of optimal execution is a very general problem in which a trader who wishes to buy or sell a {\it large\/} position $K$ of a given asset $S$---for instance wheat, shares, derivatives, etc.---is confronted with the dilemma of executing
slowly or as quick as possible. In the first case he/she would be exposed to volatility, and in the second to the laws of offer and demand.
 Thus the trader most hedge between the {\it market impact\/} (due to his trade) and the {\it volatility\/} (due to the market).

The key ingredients to study this optimization problem are: (1) The modeling of the asset---which is typically modeled as a geometric Brownian process. (2) The modeling of the so-called market impact which heuristically suggests the existence of an instantaneous impact---so-called {\it temporary\/}---and a cumulative component referred to as {\it permanent\/}. And finally (3)  one should establish a criteria of optimality.

The main {\it aim\/} of this paper is to study and characterize  the  so-called Marko\-witz-optimal open-loop execution trajectory, in terms of nonlinear second order, ordinary differential equations (Theorems \ref{thm1} and \ref{thm2}). The above is done for a  trader who wishes to buy or sell a large position $K$ of shares $S$ which evolve as a geometric Brownian motion (although in the existing literature it is considered an arithmetic Brownian motion for this problem). In this paper we only deal with deterministic strategies, also called {\it open loop\/} controls; however,  we point out---in Section \ref{cinco}---the key ingredients to address the problem with adaptive strategies, also termed {\it Markovian\/} controls (work in progress).
 
 The main {\it motivation\/} of this work is on one hand economic, but on the other the effect of {\it market impact\/} in the valuation of contingent claims, and its connection with optimal execution of {\it derivatives\/}. Intuitively, for this kind of problem, one would expect to consider adaptive strategies to tackle the questions, although  it seems natural to first understand the deterministic case.
 An important element in our analysis, will be the use of a linear stochastic differential equation, first used---to our knowledge---by Brennan and Schwartz (1980) in their study of interest rates.

 The problem of minimizing expected overall liquidity costs has been analyzed using different market models by  Bertsimas and Lo (1998), Obizhaeva and Wang, and Alfonsi et al. (2007a,2007b), just to mention a few. However, these approaches miss the volatility risk associated with time delay. Instead, Almgren and Chriss (1999,2000), suggested  studying and solving a mean-variance optimization for sales revenues in the class of deterministic strategies. Further, on Almgren and Lorenz (2007) allowed for intertemporal updating and proved that this can {\it strictly\/} improve the mean-variance performance. Nevertheless, in Schied and Sch\"oneborn (2007),  the authors study the original problem of expected utility maximization with CARA utility functions. Their main result states that for CARA investors there is surprisingly no added utility from allowing for intertemporal updating of strategies. Finally, we mention that the Hamilton-Jacobi-Bellman approach has also recently been studied in Forsyth (2009). 
 
Our paper is organized as follows: In Section 2 we introduce our model, assumptions and auxiliary results. Namely through a couple of subsequent Propositions we characterize and compute---by use of a Bre\-nnan-Schwartz type process (\ref{brennan})---the moments of  certain random variable that is relevant in our  optimization problem. Section 3 is devoted to deriving and proving the characterization of our optimal trading strategies and the optimal value function as well. After that, in Section 4, we present a couple of examples in order to exemplify the procedure derived in  Section 3. We first compare Almgren and Chriss trajectory with ours, and in Example 2 we use a temporary market impact $h$ to the power $3/5$ as suggested in the empirical study of Almgren {\it et. al.\/}'s (2005).  We conclude the paper, in Section 5,  by pointing out the key ingredients in the study of adaptive execution strategies. 
\section{Auxiliary Results}
In this section, we describe the dynamics of the asset $S$,  the so-called market impact, and we introduce the Brennan-Schwartz process. This process will allow us  not only to compute the moments of the optimization argument, but also to represent it in terms of an SDE.  For the remainder of  this section, let $c(t)$ be a fixed and differentiable function for $0\leq t\leq T<\infty$.
\subsection{The model.} Let the price of the share $S$ of a given company evolve  as a geometric process, where the random component $B$ is standard Brownian motion, 
 i.e. 
\begin{eqnarray*}
dS_t&=&S_t\left[\left(g(c(t))+\frac{dh(c(t))}{dt}\right)dt+\sigma dB_t\right],\qquad S_0=s,
\end{eqnarray*}
and where $g$ and $h$ represent  respectively the permanent---which accumulates over time---and instantaneous temporary impact. Thus, the {\it future\/} effective price per share due to our trade can be modelled as
\begin{eqnarray}\label{price}
S_t=s\exp\left\{\int_0^ug(c(v))dv+h(c(u))-\frac{1}{2}\sigma^2u+\sigma B_u\right\},
\end{eqnarray}
where $\sigma>0$ is an estimable parameter.\\[0.3cm]
\begin{remark} (a) Note that that process $\ln(S)$, where $S$ is as in (\ref{price}), coincides precisely with the ``standard'' notion of both permanent and temporary impact. That is, if we model the price process as arithmetic Brownian motion:
\begin{eqnarray*}
dS_t=\sigma dB_t+g(c(t))dt+\frac{dh}{dt}(c(t)),\qquad S_0=s
\end{eqnarray*}
then
\begin{eqnarray*}
S_t-s=\int_0^tg(c(u))du+h(c(t))+\sigma B_t.
\end{eqnarray*}
Hence, the first term in the right-hand side of the equality is {\it accumulating\/} over time, on the other hand, the second term is {\it not}.\\ (b) Next, observe that the process $c$ can be thought of as a control, which in turn may be:
\begin{enumerate}
\item an admissible process $c$ which is adapted to the natural filtration $\mathcal{F}^S$ of the associated  process (\ref{price}) is called a {\it feedback control\/},
\item an admissible process $c$ which can be written in the form $c_t=u(t,S_t)$ for some measurable map $u$ is called {\it Markovian control\/}, notice that any Markovian control is a feedback control,
\item a deterministic processes of the family of admissible controls are called {\it open loop controls\/}. 
\end{enumerate}
(c) In this paper we will only deal with open loop controls, yet to study feedback controls it will become quite useful to introduce the so-called Brennan-Schwartz process, introduced in the next subsection. The reason being, that to derive the Hamilton-Jacobi-Bellman equation---see Section 5---it is more convenient to have a diffusion instead of an average of a diffusion. 
\end{remark}
\subsection{Averaged geometric and Brennan-Schwartz processes} 
In order to study our problem we will introduce the following  averaged geometric Brownian process:
\begin{eqnarray}\label{chi}
\xi_t&:=&\int_0^tc(u)S_udu,\qquad t\geq 0.
\end{eqnarray}
By (\ref{price}) we can express $\xi_t$ as
\begin{eqnarray*}
\nonumber \xi_t&=&\int_0^t c(u)se^{\int_0^ug(c(v))dv+h(c(u))-\frac{1}{2}\sigma^2u+\sigma B_u}du.
\end{eqnarray*}
Thus, if $c(u)$ represents the number of shares bought or sold at time $u$ at price $S_u$, then $\xi_t$ represents the total amount spent or earned by the trader up to time $t$.

To compute the moments of $\xi$ we will use the following  linear non-homogeneous stochastic differential equation
\begin{eqnarray}\label{brennan}
dX_t&=&\left[c(t)-\left(g(c(t))+\frac{dh(c(t))}{dt}-\sigma^2\right)X_t\right]dt-\sigma X_tdB_t\\
\nonumber X_0&=&0,
\end{eqnarray}
which has been used for instance by Brennan and Schwartz (1980) in the modeling of interest rates, by Kawaguchi and Morimoto (2007) in environmental economics, and which may also be used to study the density of averaged geometric Brownian motion [see for instance, Linetsky (2004)].

In general, its usefulness is due to the fact that one may construct a Brennan-Schwartz process $X$ which satisfies 
$$X\stackrel{\mathcal{D}}{=}\xi,$$
where $\stackrel{\mathcal{D}}{=}$ stands for equality in distribution.
In this paper, it is used, together with It\^o's lemma to show that $\xi=S\cdot X$, which alternatively will allow us to compute the second moment of $\xi$ in terms of an iterated integral, indeed:

\begin{proposition}\label{prop1}
Let processes $S$, $\xi$, and  $X$ be as in (\ref{price}), (\ref{chi}), and  (\ref{brennan}), respectively, then
\begin{eqnarray*}
\xi_t=S_t\cdot X_t,\qquad t\geq 0
\end{eqnarray*}
and
\begin{eqnarray*}
d\xi^2_t=2\xi_tS_t c(t)dt.
\end{eqnarray*}

\end{proposition}
\begin{proof} By It\^o's lemma
\begin{eqnarray*}
d(S_t\cdot X_t)&=&S_tdX_t+X_tdS_t+dX_tdS_t\\
&=&S_tc(t)dt-S_tX_t\left(g(c(t))+\frac{dh(c(t))}{dt}\right)dt+S_tX_t\sigma^2dt\\
&&-\sigma S_tX_tdB_t+X_tS_t\left(g(c(t))+\frac{dh(c(t))}{dt}\right)dt\\
&&+\sigma X_tS_tdB_t-\sigma^2S_tX_tdt\\
&=&c(t)S_tdt\\
&=&\xi_t.
\end{eqnarray*}
Furthermore, for the second moment of $\xi$ it follows that
\begin{eqnarray*}
d\xi^2_t&=&2X_tS^2_tdX_t+2X^2_tS_tdS_t+S^2_t(dX_t)^2+X^2_t(dS_t)^2\\
&&+4X_tS_t(dX_t\cdot dS_t)\\
&=&2X_tS^2_tsc(t)dt-2X_tS^2_t\left(g(c(t))+\frac{dh(c(t))}{dt}\right)X_tdt\\
&&+2X_tS^2_t\sigma^2X_tdt-2X_tS^2_t\sigma X_tdB_t\\
&&+2X^2_tS_t\left(g(c(t))+\frac{dh(c(t))}{dt}\right)S_tdt+2X^2_tS_t\sigma S_tdB_t\\
&&+S^2_t\sigma^2X^2_tdt+X^2_t\sigma^2S^2_tdt-4X_tS_t\sigma^2X_tS_tdt\\
&=&2X_tS^2_tc(t)dt\\
&=&2\xi_t S_t c(t)dt,
\end{eqnarray*}
as claimed.
\end{proof}
\begin{remark}
The previous proposition may be derived directly from the integration by parts formula. Yet, this characterization will be useful in the study,  for instance, of the optimal trading schedule of derivatives or in the determination of Markovian controls.
\end{remark}
\subsection{Moments of $\xi$}
Now, by Proposition \ref{prop1},  it is straightforward to compute the first two moments of $\xi_t$ which will be used to solve our optimal execution problem.
\begin{corollary} Let $\xi$ be as in (\ref{chi}). Then 
\begin{eqnarray}
\label{uno}\mathbb{E}[\xi_t]&=&\int_0^t c(u)s\exp\left\{\int_0^ug(c(v))dv+h(c(u))\right\}du \\
\label{dos}\mathbb{E}[\xi^2_t]&=&2\int_0^tc(u)se^{\int_0^ug(c(n))dn+h(c(u))}\\
\nonumber&&\quad \times\left(\int_0^uc(v)se^{\int_0^vg(c(w))dw+h(c(v))+\sigma^2v}dv\right)du.
\end{eqnarray}
\end{corollary}
\begin{proof} By (\ref{chi}):
\begin{eqnarray*}
\mathbb{E}[\xi_t]&=&\int_0^tc(u)s\mathbb{E}[S_u]du\\
&=&\int_0^t c(u)s\exp\left\{\int_0^ug(c(v))dv+h(c(u))\right\}du.
\end{eqnarray*}
From Proposition \ref{prop1}:
\begin{eqnarray*}
\mathbb{E}[\xi^2_t]&=&2\mathbb{E}\left[\int_0^tc(u)sS_u\xi_udu\right]\\
&=&2\mathbb{E}\left[\int_0^tc(u)s^2S_u\left(\int_0^uc(v)sS_vdv\right)du\right]\\
\nonumber&=&2\int_0^tc(u)s^2\left(\int_0^uc(v)s\mathbb{E}[S_uS_v]dv\right)du.
\end{eqnarray*}
Therefore, since
\begin{eqnarray*}
\mathbb{E}\left[e^{\sigma B_u+\sigma B_v}\right]&=&\mathbb{E}\left[e^{\sigma(B_u-B_v)+2\sigma B_v}\right]\\
&=&\mathbb{E}\left[e^{\sigma(B_u-B_v)}\right]\mathbb{E}\left[e^{2\sigma B_v}\right]\\
&=&e^{\frac{1}{2}\sigma^2(u-v)}e^{2\sigma^2 v},
\end{eqnarray*}
it follows that
\begin{eqnarray*}
\mathbb{E}[\xi^2_t]&=&2\int_0^tc(u)se^{\int_0^ug(c(n))dn+h(c(u))}\\
&&\quad \times\left(\int_0^uc(v)se^{\int_0^vg(c(w))dw+h(c(v))+\sigma^2v}dv\right)du.
\end{eqnarray*}
\end{proof}

\section{Markowitz Optimal open-loop Trading trajectory}
In this section we derive a Markowitz-optimal open-loop trading strategy, Theorem \ref{thm1} and \ref{thm2}, employing the auxiliary results derived in the previous section. 

\subsection{Execution shortfall} If the size of the trade $K$ is ``relatively'' small  we would  expect the market impact to be negligible, that is, the trader should execute $K$ immediately. Thus, it seems natural to compare the actual total gains (losses) $\xi_T$ with the impact-free quantity $Ks$ by introducing the so-called  
execution shortfall $Y$ defined as
\begin{eqnarray}\label{short}
Y:=\xi_T-Ks.
\end{eqnarray}

If we use Markowitz optimization criterion, then our problem is equivalent to finding the trading trajectory $\{c(t)|0\leq t\leq T\}$ which minimizes
simultaneously  the expected shortfall given a fixed risk-aversion level $\lambda$ characterized by the volatility of $Y$:
\begin{eqnarray}\label{optimize}
\nonumber J[c(\cdot)]&:=&\mathbb{E}[Y]+\lambda \mathbb{V}[Y]\\
&=&\lambda\mathbb{E}[\xi^2_T]+\mathbb{E}[\xi_T]-\lambda (\mathbb{E}[\xi_T])^2-Ks.
\end{eqnarray}
In fact, if  $\lambda>0$ then (\ref{optimize}) has a unique solution, which may be represented in the following integral form: 
\begin{proposition}\label{prop3}
Suppose that the permanent impact $g$ is linear, i.e.
  $$g(x)=\alpha x,$$
for some $\alpha>0$ as suggested by Almgren {\it et. al.\/} (2005) empirical study. Let
\begin{eqnarray}\label{f}
f(x):=\int_0^x c(u)du,\qquad f'(x):=c(x)
\end{eqnarray}  
and
\begin{eqnarray*}
\gamma_1(u,f,f')&:=&\int_0^usf'(v)e^{\alpha f(v)+ h(f'(v))+\sigma^2 v}dv,\\
\gamma(u,f,f')&:=&\int_0^usf'(v)e^{\alpha f(v)+ h(f'(v))}dv.
\end{eqnarray*}
Then $J[c(\cdot)]$ in (\ref{optimize}) can be expressed as:
\begin{eqnarray*}
\int_0^T\left\{\left[2\lambda \left(\gamma_1(u)-\gamma(u)\right)+1\right]f'(u)se^{\alpha f(u)+ h(f'(u))}-\frac{\lambda Ks}{T}\right\}du.
\end{eqnarray*}
\end{proposition}
\begin{proof}
Setting
\begin{eqnarray*}
f(x):=\int_0^x c(u)du,
\end{eqnarray*}
we have $f'(x):=c(x)$. Hence, using
the integration by parts formula,
\begin{eqnarray*}
&&\int_0^tsf'(x)e^{\alpha f(x)+h(f'(x))}\left(\int_0^xsf'(y)e^{\alpha f(y)+h(f'(y))}dy\right)dx\\
&&=\left(\int_0^tsf'(x)e^{\alpha f(x)+h(f'(x))}dx\right)\left(\int_0^tsf'(x)e^{\alpha f(x)+h(f'(x))}dx\right)\\
&&\enskip -\int_0^tsf'(x)e^{\alpha f(x)+ h(f'(x))}\left(\int_0^xsf'(y)e^{\alpha f(y)+h(f'(y))}dy\right)dx
\end{eqnarray*}
implies
\begin{eqnarray}
&&\nonumber\left(\int_0^tsf'(x)e^{\alpha f(x)+h(f'(x))}dx\right)^2\\
\nonumber&&\quad=2\int_0^tsf'(x)e^{\alpha f(x)+h(f'(x))}\left(\int_0^xsf'(y)e^{\alpha f(y)+h(f'(y))}dy\right)dx\\
\nonumber &&\quad=\left(\mathbb{E}[\xi_t]\right)^2.
\end{eqnarray}
Thus, by  (\ref{uno}) and (\ref{dos}),
\begin{eqnarray*}
\gamma_1(u)&:=&\int_0^usf'(v)e^{\alpha f(v)+h(f'(v))+\sigma^2 v}dv,\\
\gamma(u)&:=&\int_0^usf'(v)e^{\alpha f(v)+h(f'(v))}dv,
\end{eqnarray*}
It follows that 
\begin{eqnarray*}
&&J[c(\cdot)]\\
&&\enskip=2\lambda\int_0^Tf'(u)se^{\alpha f(u)+ h(f'(u))}\gamma_1(u)du+\int_0^T f'(u)se^{\alpha f(u)+ h(f'(u))}du\\
&&\qquad-\int_0^T2\lambda f'(u)se^{\alpha f(u)+ h(f'(u))}\gamma(u)du- K s\\
&&\enskip=\int_0^T\left\{\left(2\lambda\gamma_1(u)+1-2\lambda\gamma(u)\right)f'(u)se^{\alpha f(u)+ h(f'(u))}-\frac{Ks}{T}\right\}du.
\end{eqnarray*}
as claimed.
\end{proof}
Observe that this last expression has the following functional form in terms of $f$:
\begin{eqnarray}\label{j}
J(f)=\int_0^t\mathcal{L}(\gamma_1(u,f,f'),\gamma(u,f,f'),f(u),f'(u))du.
\end{eqnarray}
Letting
\begin{eqnarray*}
F(f(u),f'(u)):=sf'(u)\exp\left\{\alpha f(u)+h(f'(u))\right\},
\end{eqnarray*}
we may re-express (\ref{j}) as
\begin{eqnarray*}
J(f)=\int_0^T\left(2\lambda\int_0^uF(f(v),f'(v))(e^{\sigma^2 v}-1)dv+1\right)F(f(u),f'(u))du.
\end{eqnarray*}
In particular
\begin{theorem}\label{thm1}
 Suppose that $\lambda=0$. Then, the  open-loop trading schedule $c^*$ is determined by a function $f_1$ which solves the following system
 \begin{eqnarray}\label{ee}
 \frac{\partial F}{\partial f_1}-\frac{d}{dz}\frac{\partial F}{\partial f'_1}=0,
 \end{eqnarray} 
with $f_1(0)=0$ and $f_1(T)=K$.
 \end{theorem}
 \begin{proof} If $\lambda=0$, that is, if we only wish to minimize expected execution shortfall, then:
 \begin{eqnarray*}
 J(f)=\int_0^TF(f(u),f'(u))du,
 \end{eqnarray*}
 and thus equation (\ref{ee}) follows from the Euler-Lagrange equation [see for instance Gelfand and Fomin (2000)].
  \end{proof}
  \begin{example}\label{ex1} {\rm Let $T=\alpha=1$ and  both temporary and permanent impact  be linear, i.e. $g(x)=x$, $h(x)=x$, hence:
 \begin{eqnarray*}
F(f(u),f'(u))=f'(u)\exp\left\{f(u)+f'(u)\right\}.
\end{eqnarray*} 
Thus, if one wishes to minimize the {\it expected\/} execution shortfall one should execute according to: 
\begin{eqnarray*}
f'(u)-f''(u)-(1+f'(u))(f'(u)+f''(u))=0,
\end{eqnarray*}
given that $f(0)=0$, and $f(1)=K$. }
\end{example}
\begin{remark} Note that as $\lambda$ increases the client is willing to be more exposed to risk. This idea is equivalent to saying that he/she is willing to increase the speed of execution. Hence we have that for $\lambda>0$, the optimal trajectory will dominate $f_1$:
$$f(s)\geq f_1(s),\qquad 0\leq s\leq T$$
In other words we may decompose $f(s)=f_1(s)+f_2(s)$. This last observation together with our constraint $f(0)=0$, $f(T)=K$, lead to the following two facts:
\begin{eqnarray}
f_2(t)&\geq& 0,\qquad 0\leq t\leq T\\
f_2(0)=f_2(T)&=&0.
\end{eqnarray}
\end{remark}
\begin{remark} The previous remark and equation (\ref{j}) suggest a 2 step procedure to find the optimal trajectory $f$. Namely, first find $f_1$ and given that information solve for $f=f_1+f_2$
\end{remark}
\begin{theorem}\label{thm2} The optimal differentiable trajectory $f$ which solves (\ref{j}) is given by $f_1+f_2$, where $f_1$ is given in Theorem 3.3 and $f_2$ satisfies for $0\leq v\leq T$:
\begin{eqnarray*}
&&f_2(u)\left(2\lambda\int_0^uF^1(f(v),f'(v))dv+1\right)\cdot\left[\frac{\partial F}{\partial f}-\frac{d}{du}\left(\frac{\partial F}{\partial f'}\right)\right]\\
&&\qquad +2\lambda\int_0^uf_2(v)\left[\frac{\partial F^1}{\partial f}-\frac{d}{dv}\left(\frac{\partial F^1}{\partial f'}\right)\right]dv\cdot F(f(u),f'(u))=0,
\end{eqnarray*}
where $f_2(0)=f_2(T)=0$.
\end{theorem}
\begin{proof} The idea is to follow the derivation of the Euler-Lagrange equation [see for instance, Gelfand and Fomin (2000)], but in this case, the unknown function $f_2$ will play the role of the perturbation. Thus, it is essential the fact that $f_2(0)=f_2(T)=0$.

Let
\begin{eqnarray*}
f(v)&=&f_1(v)+f_2(v)\\
g_\epsilon(v)&=&f_1(v)+\epsilon f_2(v)
\end{eqnarray*} 
where $f_2(0)=0=f_2(T)$
and
\begin{eqnarray*}
J(\epsilon)=\int_0^T\left(2\lambda\int_0^uF(g_\epsilon(v),g_\epsilon'(v))(e^{\sigma^2v}-1)dv+1\right)F(g_\epsilon(u),g_\epsilon'(u))du
\end{eqnarray*}
then
\begin{eqnarray*}
\frac{dJ}{d\epsilon}(\epsilon)&=&\int_0^T\left(2\lambda\int_0^u\frac{dF}{d\epsilon}(g_\epsilon(v),g_\epsilon'(v))(e^{\sigma^2v}-1)dv\right)\\
&&\enskip\times F(g_\epsilon(u),g_\epsilon'(u))du\\
&&+\int_0^T\left(2\lambda\int_0^uF(g_\epsilon(v),g_\epsilon'(v))(e^{\sigma^2v}-1)dv+1\right)\\
&&\enskip\times\frac{dF}{d\epsilon}(g_\epsilon(u),g_\epsilon'(u))du\\
&=&\int_0^T\left(2\lambda\int_0^u\left(f_2(v)\frac{\partial F}{\partial g_\epsilon}+f_2'(v)\frac{\partial F}{\partial g_\epsilon'}\right)(e^{\sigma^2v}-1)dv\right)\\
&&\enskip\times F(g_\epsilon(u),g_\epsilon'(u))du\\
&&+\int_0^T\left(2\lambda\int_0^uF(g_\epsilon(v),g_\epsilon'(v))(e^{\sigma^2v}-1)dv+1\right)\\
&&\enskip\times\left(f_2(u)\frac{\partial F}{\partial g_\epsilon}+f_2'(u)\frac{\partial F}{\partial g_\epsilon'}\right)du
\end{eqnarray*}
if we set
\begin{eqnarray*}
\gamma_\epsilon(u)=2\lambda\int_0^uF(g_\epsilon(v),g_\epsilon'(v))(e^{\sigma^2v}-1)dv+1
\end{eqnarray*}
and by the integration by parts formula we have that
\begin{eqnarray*}
\frac{dJ}{d\epsilon}(\epsilon)&=&\int_0^T\Bigg{(}2\lambda\int_0^uf_2(v)(e^{\sigma^2v}-1)\left[\frac{\partial F}{\partial g_\epsilon}-\frac{d}{dv}\left\{\frac{\partial F}{\partial g_\epsilon'}\right\}\right]dv\\
&&\quad-2\lambda\int_0^uf_2(v)\frac{\partial F}{\partial g_\epsilon'}\sigma^2e^{\sigma^2v}dv\\
&&\quad+2\lambda f_2(u)\frac{\partial F}{\partial g_\epsilon'}(e^{\sigma^2u}-1)\Bigg{)}F(g_\epsilon(u),g_\epsilon'(u))du\\
&&+\int_0^T\gamma_\epsilon(u)f_2(u)\left[\frac{\partial F}{\partial g_\epsilon}-\frac{d}{du}\left(\frac{\partial F}{\partial g_\epsilon'}\right)\right]du\\
&&-\int_0^Tf_2(u)\frac{\partial F}{\partial g_\epsilon'}\frac{d}{du}\left(\gamma_\epsilon(u)\right)du
\end{eqnarray*}
Let us compute
\begin{eqnarray*}
\frac{d}{du}(\gamma_\epsilon(u))=2\lambda F(g_\epsilon(u),g_\epsilon'(u))(e^{\sigma^2u}-1)
\end{eqnarray*}
which yields
\begin{eqnarray*}
J'(\epsilon)&=&\int_0^T\Bigg{(}2\lambda\int_0^uf_2(v)(e^{\sigma^2v}-1)\left[\frac{\partial F}{\partial g_\epsilon}-\frac{d}{dv}\left\{\frac{\partial F}{\partial g_\epsilon'}\right\}\right]dv\\
&&\qquad-2\lambda\int_0^uf_2(v)\frac{\partial F}{\partial g_\epsilon'}\sigma^2e^{\sigma^2v}dv\Bigg{)} F(g_\epsilon(u),g_\epsilon'(u))du\\
&&+\int_0^T\gamma_\epsilon(u)f_2(u)\left[\frac{\partial F}{\partial g_\epsilon}-\frac{d}{du}\left(\frac{\partial F}{\partial g_\epsilon'}\right)\right]du.
\end{eqnarray*}
But observe that when $J'(1)$ we have
\begin{eqnarray*}
J'(1)&=&\int_0^T\Bigg{(}2\lambda\int_0^uf_2(v)(e^{\sigma^2v}-1)\left[\frac{\partial F}{\partial f}-\frac{d}{dv}\left\{\frac{\partial F}{\partial f'}\right\}\right]dv\\
&&\qquad-2\lambda\int_0^uf_2(v)\frac{\partial F}{\partial f'}\sigma^2e^{\sigma^2v}dv\Bigg{)} F(f(u),f'(u))du\\
&&+\int_0^T\gamma_1(u)f_2(u)\left[\frac{\partial F}{\partial f}-\frac{d}{du}\left(\frac{\partial F}{\partial f'}\right)\right]du\\
&=&0.
\end{eqnarray*}
Now, given that $f_2(0)=f_2(T)=0$ and from the fundamental lemma of Calculus of variations we have that:
\begin{eqnarray*}
&&f_2(u)\gamma_1(u)\left[\frac{\partial F}{\partial f}-\frac{d}{du}\left(\frac{\partial F}{\partial f'}\right)\right]\\
&&\qquad +2\lambda\int_0^uf_2(v)\left[\frac{\partial F^1}{\partial f}-\frac{d}{dv}\left(\frac{\partial F^1}{\partial f'}\right)\right]dv\cdot F(f(u),f'(u))=0
\end{eqnarray*}
with the constraint that $f_2(u)>0$ for $u\in(0,T)$ and $f_2(T)=f_2(0)=0$.
or
\begin{eqnarray*}
&&f_2(u)\left(2\lambda\int_0^uF^1(f(v),f'(v))dv+1\right)\cdot\left[\frac{\partial F}{\partial f}-\frac{d}{du}\left(\frac{\partial F}{\partial f'}\right)\right]\\
&&\qquad +2\lambda\int_0^uf_2(v)\left[\frac{\partial F^1}{\partial f}-\frac{d}{dv}\left(\frac{\partial F^1}{\partial f'}\right)\right]dv\cdot F(f(u),f'(u))=0
\end{eqnarray*}
as claimed.
\end{proof}
\begin{example}\label{ex2} ({\it cont\/}.) {\rm With linear and temporary impact as before, that is $g(x)=x$ and $h(x)=x$. You may find the optimal trading trajectory $f$ for arbitrary $\lambda\geq 0$ by first determining $f_1$, using Theorem \ref{thm1},  next you determine $f_2$ by use of Theorem \ref{thm2}. That is, letting:
\begin{eqnarray*}
F(f(v),f'(v))&=&f'(v)\exp\left\{f(v)+f'(v)\right\}\\
F^1(f(v),f'(v))&=&F(f(v),f'(v))(e^{\sigma^2v}-1)
\end{eqnarray*}
Theorem \ref{thm2} states that $f_2$ satisfies the following identity
\begin{eqnarray*}
&&-f_2(u)\left\{2f''(u)+f'(u)\left[f'(u)+f''(u)\right]\right\}\\
&&\qquad\times\left(2\lambda\int_0^uF^1(f(v),f'(v))dv+1\right)\\
&&\qquad+ 2\lambda f'(u)\int_0^uf_2(u)\left[\frac{\partial F^1}{\partial f}-\frac{d}{dv}\left(\frac{\partial F^1}{\partial f'}\right)\right]dv=0
\end{eqnarray*}
where
\begin{eqnarray*}
&&\frac{\partial F^1}{\partial f}-\frac{d}{du}\left(\frac{\partial F^1}{\partial f'}\right)\\
&&\quad=e^{f(u)+f'(u)}\\
&&\qquad\times-\left[\left\{2f''(u)+f'(u)(f'(u)+f''(u))\right\}(e^u-1)+(1+f'(u))e^u\right].
\end{eqnarray*} }
\end{example}

\section{Examples}
\subsection{Example} The first example we want to study is the case in which both the temporary and the permanent impact are linear as in Examples \ref{ex1} and \ref{ex2}. The motivation is to compare our results with those obtained by Almgren and Chriss (2000). In this example we have set $T=1$ and $K=3$. The solutions have been numerically calculated using Theorems \ref{thm1} and \ref{thm2} and then plotted in Figures 1 and 2. 

It may be observed that, as one would expect, the solutions are not the same and in fact---under the present conditions---our strategy dominates Almgren and Chriss. 

\subsection{Example} The next example we want to study is the case in which the permanent impact is some power less than 1. Namely $h(x)=x^{3/5}$, and the permanent is linear. First from Theorem \ref{thm1} we have that
\begin{eqnarray*}
F(f_1(u),f'_1(u)):=sf'_1(u)\exp\left\{f_1(u)+(f'_1(u))^{3/5}\right\}
\end{eqnarray*}
and $f_1$ is the solution to:
\begin{eqnarray*}
f'_1-\left(f'_1+\frac{3}{3}\frac{f''_1}{(f'_1)^{2/5}}\right)-\left(\frac{3}{5}\right)^2\frac{f''_1}{(f'_1)^{2/5}}-\frac{3}{5}(f'_1)^{3/5}\left(f'_1+\frac{3}{5}\frac{f''_1}{(f'_1)^{2/5}}\right)=0,
\end{eqnarray*}
$f_1(0)=0$ and $f_1(T)=K$. In particular, with $T=1$ and $K=3$ as in the previous example we have plotted our result in Figure 3. Note that the sublinear impact has increased the speed of execution. Next, we computed the Markowitz-optimal open-loop trajectory, with $\lambda=1$ by first computing $f_2$ as described in Theorem \ref{thm2}.
\begin{remark} The previous exercise suggests that if one chooses the temporary impact to be sub-linear, the solution---with all the other parameters fixed---will always dominate its linear counterpart. On the other hand, if the temporary impact is super-linear, the solution will be dominated by its linear counterpart.  

A natural question is: What is the correct form of $h$ given our model?
\end{remark}
\section{Remarks on Markovian controls}\label{cinco}

As pointed out in the introduction, we have only dealt with {\it differentiable\/} deterministic controls---also known as open loop controls. Furthermore, our criteria of optimal is in the {\it Markowitz\/} sense. A couple of natural and reasonable question arise: how can we study  {\it feedback controls\/}? How can we optimize with respect to general utility functions? Namely, given some utility function $U$ we want to find a trajectory $c$ such that
\begin{eqnarray*}
\sup_{c\in\mathcal{U}}\mathbb{E}_{t,x}[U(Y)]
\end{eqnarray*}
where $\mathcal{U}$ is the set of admissible controls and $Y=\xi-KT$ is the execution shortfall. But $\xi$ is in general a very difficult creature to characterize, unless you {\bf observe} that you may construct a diffusion with the following dynamics:
\begin{eqnarray*}
dX_t=\left(c_t+\left[g(c_t)+\frac{dh}{dt}(c_t)\right]X_t\right)dt+\sigma X_tdB_t,\qquad X_0=0
\end{eqnarray*} 
that is equal in distribution to $\xi$, i.e.
\begin{eqnarray*}
\xi_t\stackrel{\mathcal{D}}{=}X_t,\qquad \forall t.
\end{eqnarray*}
Thus
\begin{eqnarray*}
\sup_{c\in\mathcal{U}}\mathbb{E}_{t,x}[U(Y)]=\sup_{c\in\mathcal{U}}\mathbb{E}_{t,x}[U(X_T-KT)]
\end{eqnarray*}
and now you may proceed to derive the Hamilton-Jacobi-Bellman equation. It is precisely this question, which the authors are investigating presently.

\begin{figure}\label{fig1}
\hspace{-1.3 cm}\includegraphics[scale=.5,height=12cm,width=12cm]{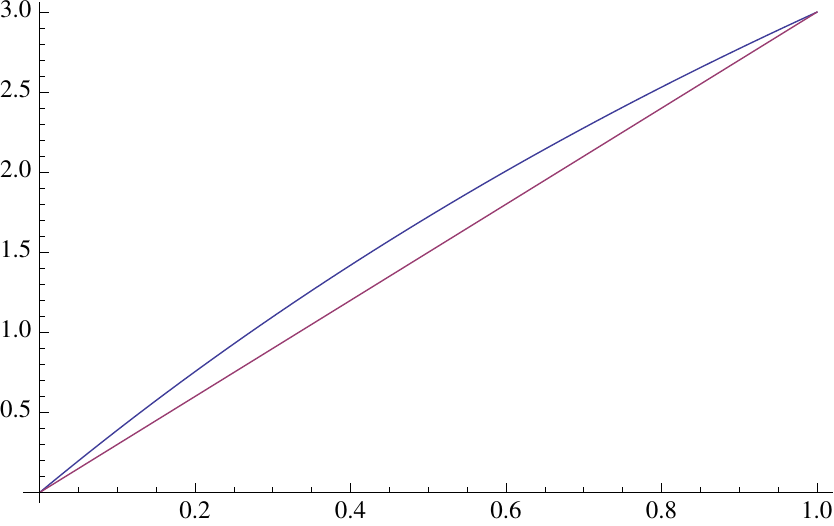}
\caption{The graph is plotted with Mathematica. For $\lambda=0$ the blue line (upper line) represents our optimal trading trajectory, the red line is Almgren and Chriss.}
\end{figure}

\begin{figure}\label{fig2}
\hspace{-1.3 cm}\includegraphics[scale=.5,height=12cm,width=12cm]{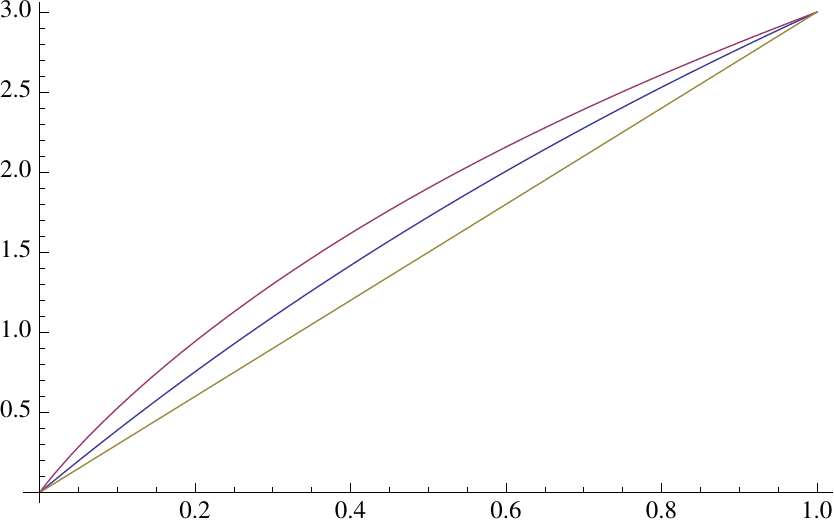}
\caption{The graph is plotted with Mathematica. The upper red line is with $\lambda=1$, the middle blue line is with $\lambda=0$ and the lower yellow line is the case of arithmetic Brownian motion.}
\end{figure}

\begin{figure}
\hspace{-1.3 cm}\includegraphics[scale=.5,height=12cm,width=12cm]{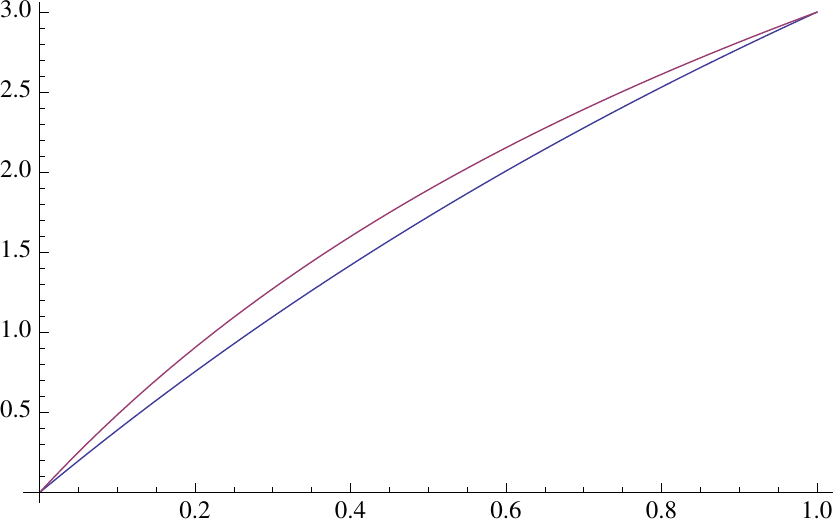}
\caption{The graph is plotted with Mathematica. The upper red line represents the trading trajectory with sub-linear temporary impact. The lower blue line represents the trajectory with linear temporary impact. }
\end{figure}

\begin{figure}
\hspace{-1.3 cm}\includegraphics[scale=.5,height=12cm,width=12cm]{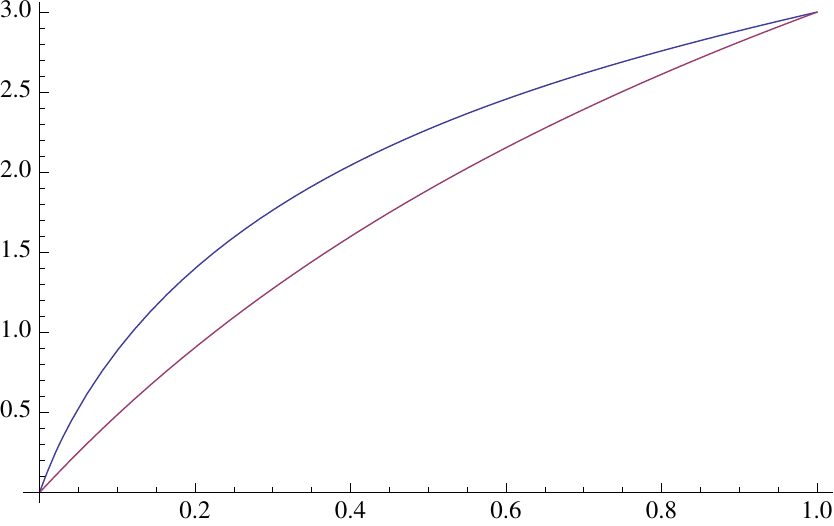}
\caption{The graph is plotted with Mathematica. The upper blue line represents the trading trajectory with sub-linear temporary impact and $\lambda=1$. The lower red line represents the trajectory with sub-linear temporary impact and $\lambda=0$. }
\end{figure}

\end{document}